\documentclass[runningheads]{llncs}
\usepackage{graphicx}
\usepackage{amsmath}
\usepackage{multirow}
\usepackage{marvosym}
\begin{document}
\title{Multi-votes Election Control by Selecting Rules}
%
%
\author{Fengbo Wang \inst{1,2} \and
Aizhong Zhou \inst{1,3} \textsuperscript{(\Letter)} \and
Jianliang Xu \inst{1,4}}
\authorrunning{F. Wang et al.}
%
\institute{Ocean University of China, Qingdao, Shandong, China \and
\email{wfb@stu.ouc.edu.cn} \and
\email{zhouaizhong@ouc.edu.cn} \and
\email{xjl9898@ouc.edu.cn}}
\maketitle              
\begin{abstract}
We study the election control problem with multi-votes, whe\-re each voter can present a single vote according different views (or layers, we use ``layer" to represent ``view").
For example, according to the attributes of candidates, such as: education, hobby or the relationship of candidates, a voter may present different preferences for the same candidate set. Here, we consider a new model of election control that by assigning different rules to the votes from different layers, makes the special candidate p being the winner of the election (a rule can be assigned to different layers). 
Assuming a set of candidates C among a special candidate ``p", a set of voters V, and t layers, each voter gives t votes over all candidates, one for each layer, a set of voting rules R, the task is to find an assignment of rules to each layer that p is acceptable for voters (possible winner of the election). 
Three models are considered (denoted as sum-model, max-model, and min-model) to measure the satisfaction of each voter. 
In this paper, we analyze the computational complexity of finding such a rule assignment, including classical complexity and parameterized complexity. 
It is interesting to find out that 1) it is NP-hard even if there are only two voters in the sum-model, or there are only two rules in sum-model and max-model; 2) it is intractable with the number of layers as parameter for all of three models; 3) even the satisfaction of each vote is set as dichotomous, 1 or 0, it remains hard to find out an acceptable rule assignment. 
Furthermore, we also get some other intractable and tractable results.

\end{abstract}
\section{Introduction}

Elections are a commonly used mechanism to achieve preference aggregation and have applications in multi-agent settings and political domains. 
This problem also plays a fundamental role in artificial intelligence and social choice~\cite{ref_article1,ref_article2}.
Most cases studied are set to find out a single winner, voting can also be used to select a fixed-size set of winners (multi-winner), called committee.

The first innovation of our work is that we consider the condition where each voter can present multi-votes.
The traditional election only allows each voter to provide a single vote, which is insufficient in many real applications.
For example, a single person has different attributes in different scenes, such as, he is a husband when he accompanies with his wife, and is a teacher when he faces to the students.
It is natural for us to present different preferences among agents from the different viewpoints.
Such as, from a romantic perspective, giving rose to others is better than a candy, while from a perspective of filling the stomach, we often prefer candy to rose.
The conditions where each voter is allowed to present multi-votes studied in Aziz et al.~\cite{ref_article3},
Chen et al.~\cite{ref_article4},
Miyazaki and Okamoto~\cite{ref_article5} and Robert Bredereck et al.~\cite{ref_article25}.
Wen et al.~\cite{ref_article6} studied the multi-preference model of matching. 
A related work of Jain and Talmon~\cite{ref_article7} studied committee selection with multi-modal preferences, which assuming a set of candidates $A$, a set of voters $V$, and $\ell$ layers,
where each voter $v\in V$ has ordinal preferences over the alternatives for each layer separately, the task is to select an acceptable committee $S\subset A$ of size $k$.

We also consider the election with uncertainty, which is another hot topic in the research of social choice. 
In the context of winner determination, perhaps the most prominent problem in this category is vote uncertainty, the possible/necessary winner problem~\cite{ref_article8}, where the voting rule is public information, but for each voter, only a partial order over the candidates is known; the goal is to determine if a candidate wins the election for some way (the possible winner) or for every way (the necessary winner) of completing the voters’ preferences; a probabilistic variant of this problem has been considered~\cite{ref_article9}.
Kocot considered if there is a committee that meets or exceeds the respective lower bound with respect to each of the rules~\cite{ref_article10}.
Uncertainty about the voting rules has been recently investigated by Baumeister et al.~\cite{ref_article11}, who also consider the situation where the voting rule will be chosen from a fixed set.
Maciej Kocot et al.~\cite{ref_article10} has studied winner determination and voting manipulation under uncertainty.
Edith Elkind and G\'{a}bor Erd\'{e}lyi~\cite{ref_article12} studied the complexity of manipulation for the setting where there is uncertainty about the voting rule: the manipulator(s) know that the election will be conducted using a voting rule from a given list, and need to select their votes so as to succeed no matter which voting rule will eventually be chosen.
A similar work has been in Conitzer et al~\cite{ref_article13}.
We follow this line and continue to consider the scene where the voting rules are uncertain, and our work is to find a set of satisfying rules assigned to each layer.

Another contribution of this paper is that we consider a new model of election control where assigning rules to each layer to determine the election winner (the satisfaction of the vote is achieved by the assigned rule).
It can be seen as an attack to control the winner of the election.
The computational complexity of elections under attacks has been studied extensively, since Bartholdi et al.~\cite{ref_article14} introduced the usage of computational complexity as a barrier to protect elections against different manipulative actions.
The common attacks among manipulation, control, and bribery.
See the book chapters~\cite{ref_article15,ref_article16} for recent surveys of related results.
Here, we focus on the control attacks on elections, where an election chair attempts by doing some operations to make a special candidate win the election, the {\em constructive control} model~\cite{ref_article14}, or lose the election, the {\em destructive control} model~\cite{ref_article17}.
The common operations include adding candidates, adding votes, deleting candidates, deleting votes, partition of candidates, and partition of votes et al.
Complexity results of control problems have been obtained for many election systems such as Plurality,
Condorcet, Approval Voting, Copeland, Schulze Voting and Borda~\cite{ref_article15,ref_article18,ref_article19,ref_article20}.
Operations such as partitioning the candidates or votes have also been consider~\cite{ref_article20,ref_article21}.
Furthermore, control problems have also been studied in connection to some special vote structures such as single-peaked or single-dived~\cite{ref_article22}.
In this paper, we consider another operation that by selecting rules to different votes to make a special candidate being the winner of the election (the constructive control case). 
Considering the computational complexity of this new model of election control can reduce the impact on the fairness of election and ensure the rationality of the winner. 

In summary, our work combines the characteristic of multi-votes, uncertainty and control together, and studies the computational complexity of the problem, called \emph{Multi-votes Election Control By Selecting Rules, MECSR} for short.
The multi-votes election provides rule uncertainty with a well existence opportunity.
When each voter provides a single vote, only a rule is applied to this voter. So, the task is to chose a rule from the rule set or to determine which vote is applied with the chosen rule. 
However, when each voter provides multi-votes, it presents a possibility that multi-votes are applied with different rules, and the task is to achieve a satisfying rule assignment (the satisfaction of voters is enough) to different rules. 
To measure the satisfaction of each voter, we consider three models, \emph{Sum-Model}, \emph{Max-Model}, and \emph{Min-Model}, and find out that the {\sc MECSR} problem is NP-hard for all of the three models.
We continue to study the parameterized complexity with the three models, and get some tractable and intractable results (shown in table.1).
It is interesting to find out that 1) it is NP-hard even if there are only two voters in sum-model, or two rules in sum-model or max-model, 2) it is intractable with the number of layers as parameter for all of the three rules, 3) even the satisfaction of each vote is set as dichotomous, 1 or 0, it is still hard to find out an acceptable rule assignment.
In the following of this paper, we first present the preliminaries in section 2, and show the details of classical and parameterized complexity results in section 3.
Finally, we summarize our work and present some interesting future work in  section 4.

\begin{table*}
\label{Results}
\centering\caption{In this table, we summarize our results including classical and parameterized complexity. The $n$ denotes the number of voters, $t$ denotes the number of votes presented by each voter (the number of layers), $\ell$ denotes the number of rules, and $\alpha$ denotes the number of satisfied voters. It is trivially in P when there is only one voter, one layer, or one rule. And when all of the voters all voters accept the winner with min-model ($n=\alpha$), we can obviously find out an acceptable rule assignment in polynomial-time. Therefore, it is FPT of {\sc MECSR} problem with the number of voters as parameter. All of the results shown in this table are reached, even if the satisfaction of each vote is set as dichotomous, $0$ or $1$.}
\resizebox{\textwidth}{!}{
\begin{tabular}{|c|c|c|c|c|c|c|}
\hline
{Model} & {Classical Complexity}
& \multicolumn{4}{c|}{Parameterized Complexity}\\
\hline
{} &{}& {$n$} &{$t$}
&{$\ell$}&{$\alpha$}\\
\hline
\multirow{3}{*}{Sum} &\multirow{2}{*}{\bf NP-hard} & {\bf $n\leq t:$ W[2]-hard } & \multirow{2}{*}{\bf W[2]-hard} & \multirow{2}{*}{\bf $\ell=2$:NP-hard} &\multirow{2}{*}{\bf W[1]-hard}\\
&{}& {\bf $n>t:$ FPT}&{}& {}&{}\\
&{\bf [Theorem~\ref{Sum-t}]}&{\bf [Theorem~\ref{sum-n}]}&{\bf [Theorem~\ref{Sum-t}]}&{\bf [Theorem~\ref{Sum-n and l}]}&{\bf [Theorem~\ref{sum-alpha}]}\\
\hline
\multirow{3}{*}{Max} &\multirow{2}{*}{\bf NP-hard} & {\bf $n\leq t:$ W[2]-hard } & \multirow{2}{*}{\bf W[2]-hard} & \multirow{2}{*}{\bf $\ell=2$:NP-hard} &\multirow{2}{*}{\bf W[1]-hard}\\
&{}& {\bf $n>t:$ FPT}&{}& {}&{}\\
&{\bf [Theorem~\ref{max-t}]}&{\bf [Theorem~\ref{max-n}]}&{\bf [Theorem~\ref{max-t}]}&{\bf [Theorem~\ref{max-l}]}&{\bf [Theorem~\ref{max-alpha}]}\\
\hline
\multirow{3}{*}{Min} &\multirow{2}{*}{\bf NP-hard} & \multirow{3}{*}{FPT} & \multirow{2}{*}{\bf W[1]-hard} &{\bf $\ell\leq t:$ W[1]-hard} & \multirow{2}{*}{\bf W[1]-hard}\\
{}&{}&{}&{}&{\bf $\ell> t:$ FPT}&{}\\
&{\bf [Theorem~\ref{min t-alpha}]}&{}&{\bf [Theorem~\ref{min t-alpha}]}&{\bf [Theorem~\ref{min-ell}]}&{\bf [Theorem~\ref{min t-alpha}]}\\
\hline
\end{tabular}
}
\label{table:truncated}
\end{table*}

\section{Preliminaries}
A traditional election denoted as $E=(C,V)$ among $m$ candidates in $C$ and $n$ votes in $V$ from $n$ voters.
The aim of the election is to select a single satisfied candidate $c$ from $C$ to be the winner, according to the votes in $V$.
Here, we analyse a special model of election where each voter gives $t$ votes over all candidates $C$ from in $t$ layers, such as experience or education, with each vote corresponding to one layer.
The vote set $V$ contains $n$ subset $V_i(1\leq i \leq n)$, where each subset corresponds to a voter, $V=\bigcup_{1\leq i\leq n} V_i$; and each subset $V_i$ contains $t$ votes $v_i^j(1\leq j\leq t)$, each vote corresponds to a layer of the voter, $V_i=\bigcup_{1\leq j\leq t} v_i^j$. 
The vote $v_i^j$ is presented by the $i-$th voter from the $j-$th layer.
To measure the satisfaction of vote $v$ with the chosen winner $c$ being the winner, we often think about the rules (such as: Borda) which can calculate a value, ${\tt Sat}(c,v,r)$ with rule.
The satisfaction of a voter ${\tt Sat}(V)$ is obtained according to the $t$ vote satisfactions ${\tt Sat}(c,v,r)$ with $v \in V$. 
When Sat(V) reaches a given threshold d, we call the voter \textit{accepts} the winner $c$, otherwise, we call the voter \textit{rejects} the winner $c$.
The satisfaction of each voter is determined by combination of his $t$ votes and the chosen winner $c$. 
Hereby, we consider the condition called as \emph{Multi-votes Election Control By Selecting Rules(MECSR)}  that given a set of rules $R=\{r_1,\cdots,r_{\ell}\}$, the satisfaction threshold $d$ and a special candidate $p\in C$, is there an assignment of rules to each layer to make sure that the satisfaction of each voter is at least $d$ with $p$ being the winner?
There are some notes about our work described as follows:
1). Same layers of different voters share common rules, and a single rule can be assigned to different layers; 
2). We do not require $p$ to be the unique winner,
which means the rule assignment may potentially result in another candidate being the winner, such an outcome is acceptable; 
3). Here, we consider the rules which can calculate ${\tt Sat}(c,v)$ in polynomial time with given candidate $c$ and vote $v$;   
4). Although we just consider the special candidate $p$ being the winner here, our work can be applied to the committee election with the rules which can get a value of each candidate from the votes directly, such as Borda, plurality or veto.

\subsection{Problem Definition}
Here, We define the central problem of this paper.
\begin{quote}
    {\bf Multi-votes Election Control By Selecting Rules(MECSR)}\\
    {\bf Input}: An election $E=(C, V, R, t)$, each voter provides $t$ votes over all candidates in $C$ where each vote for each layer, a set $R$ of rules, a special candidate $p \in C$, and two positive integers $d$ and $\alpha$.\\
    {\bf Question}: Is there an assignment of rules in $R$ for each layer that the at least $\alpha$ voters accept the winner $p$ (${\tt Sat}(V)\geq d$)?
\end{quote}
Since the special candidate $p\in C$, each vote $v\in V$, and each rule $r\in R$ are part of the input, we can calculate ${\tt Sat}(p,v,r)$ in polynomial time, denoted as ${\tt Sat}(v,r)$.
Therefore, we can consider ${\tt Sat}(v,r)$ as part of the input directly, without specifying the formats of each vote $v$ and each rule $r$.

In this paper, we investigate three models of calculating voter satisfaction, which have also been studied by Aziz et al. \cite{ref_article3}: 
\begin{itemize}
    \item Sum-Model: The satisfaction of each voter is the total satisfaction of all $t$ layers, ${\tt Sat}(V_i)=\sum_{j=1}^{t} {\tt Sat}(v_i^j)$;
    \item Max-Model: The satisfaction of each voter is the maximal satisfaction among the $t$ layers, ${\tt Sat}(V_i)=\max\{{\tt Sat}(v_i^j)| 1\leq j\leq t\}$;
    \item Min-Model: The satisfaction of each voter is the minimal satisfaction among the $t$ layers, ${\tt Sat}(V_i)=\min\{{\tt Sat}(v_i^j)| 1\leq j\leq t\}$.
\end{itemize}
The sum-model measures the total satisfaction of a voter and does not consider the individual satisfaction of each vote; in the max-model, voters accept the chosen candidate when the satisfaction is enough from at least one vote; and for the min-model, voters accept the chosen candidate only if the satisfaction is enough from all votes.

\subsection{Parameterized Complexity}
Parameterized complexity allows us to give a more refined analysis of computational problems and in particular,
can provide a deep exploration of the connection between the problem complexity and various problem-specific parameters.
A fixed-parameter tractable (FPT) problem admits an $O(f(k) \cdot |I|^{O(1)})$-time algorithm, where~$I$ denotes the
whole input instance, $k$ is the parameter, and~$f$ can be any computable function.
Fixed-parameter intractable problems can be classified into many complexity classes, where the most fundamental ones are W[1] and W[2].
A problem is para-NP-hard with respect to parameter $k$, when the problem remains NP-hard even if $k$ is a  fixed constant.
For more details on parameterized complexity, we refer to ~\cite{ref_book23,ref_book24}.

\section{Classical and parameterized complexity}
In this section, we show the computational complexity of {\sc MECSR} problem with sum-model, max-model, or min-model.
It is trivial {\sc MECSR} problem is in P when there is only one rule, one layer, or one voter.
Otherwise, {\sc MECSR} problem is NP-hard for all of the three models.
Furthermore, we achieve some intractable results and two tractable results.
For ease of the description, we use $j\in[n]$ to represent $1\leq j\leq n$, use $j\in [n_1,n_2]$ to represent $n_1\leq j\leq n_2$,  where $j$ is a non-negative integer; use $N[v]$ to represent the neighborhood set of $v$, which includes the vertex $v$ itself along with its adjacent vertices.

\subsection{Complexity with Sum model}
In this section, we present the complexity results for the {\sc MECSR} problem with sum-model. 
In sum-model, the satisfaction of each voter is calculated by the sum of satisfactions from all $t$ layers, denoted as ${\tt Sat}(V_i) = \sum_{j=1}^t {\tt Sat}(v_{i}^j)$. 
We achieve 1) {\sc MECSR} problem with sum-model is NP-hard even when there are only two rules; 2) it is W[2]-hard with respect to the number of layers $t$; 3)it is W[1]-hard with the number of satisfied voters $\alpha$ as the parameter.

\begin{theorem}
\label{Sum-t}
The {\sc MECSR} problem with sum-model is NP-hard and is W[2]-hard with respect to the number of layers $t$.
\end{theorem}
\begin{proof}
We prove the theorem by reducing from {\sc Dominating Set} problem, which given a graph~$\mathcal{G}=(\mathcal{V},\mathcal{E})$ and an integer~$k'$ where $|\mathcal{V}|=m'$ and $|\mathcal{E}|=n'$, asks for a size-$\leq k'$ vertex subset $\mathcal{V'} \subseteq \mathcal{V}$ where $\forall v\in \mathcal{V}, \exists v' \in \mathcal{V'}, v\in N[v']$.
It is known that {\sc Dominating set} problem is NP-hard and is W[2]-hard with respect to the size of $\mathcal{V'}$~\cite{ref_book24}.
We construct an {\sc MECSR} instance~$(E=(C,V,R,t), \alpha, d)$ from~$(\mathcal{G}=(\mathcal{V},\mathcal{E}), k')$ as follows.

For each vertex $v_i\in \mathcal{V}, i\in [m']$, we construct a voter $V_i$ and a rule $r_i$, $V=\bigcup_{i=1}^{m'} V_i$, $R=\bigcup_{i=1}^{m'} r_i$.
There are $k'$ layers in total, $t=k'$.
For each voter $V_i, i\in [m']$, we construct $k'$ votes, $V_i=\bigcup_{j=1}^{k'} v_i^j$.
The satisfaction of vote $v_i^j (j\in [k'])$ with rule $r_k$ is set to 1, if the corresponding vertices $v_i$ and $v_k$ satisfy $v_i\in N[v_k]$ in $\mathcal{G}$, ${\tt Sat}(v_i^j,r_k)=1$; otherwise, the satisfaction is set to 0, ${\tt Sat}(v_i^j,r_k)=0$.
$$
{\tt Sat}(v_i^j,r_k)=\left\{
\begin{aligned}
&&1 &,& v_i\in N[v_k], \\
&&0 &,& v_i\notin N[v_k]. \\
\end{aligned}
\right.
$$
Note that all $t=k'$ votes of one voter are the same.
Let $d:=1$, $\alpha:=n'$.
Now we prove that there is a size-$k'$ dominating set in $\mathcal{G}$ if and only if there is a rule assignment solution of {\sc MECSR} problem with sum-model.

``$\Longrightarrow$'': If there is a size-$\leq k'$ dominating set $DS$ in $\mathcal{G}$, $|DS|\leq k'$ and $\forall v\in \mathcal{V}, \exists v'\in DS, v\in N[v']$.
Let $R'$ be the set of rules corresponding to the vertices in $DS$, $|R'| \leq k'$, that is, $\forall v_i\in DS, r_i\in R'$.
For each vertex $v_i\in DS$, $r_i$ is assigned to one layer.
It means each rule in $R'$ is assigned to one layer.
Since all the $k'$ votes of one voter are same, the assignment order of the chosen $k'$ rules has no effect on the satisfaction of each voter.
Since each vertex $v_i \in \mathcal{V}$ is adjacent to at least one vertex $v_{k} \in DS$, it means for each voter the satisfaction of at least one layer with the rule $r_k$ is 1, that is $\exists r_k\in R', {\tt Sat}(p,v_i^j,r_k)=1$.
So, for each voter $V_i$, the total satisfaction of $V_i$ is at least $d=1$:
\begin{equation*}
\begin{split}
     {\tt Sat}(V_i)=\sum_{j=1}^{k'}{\tt Sat}(v_i^j,r_{j'})\geq {\tt Sat}(v_i^j,r_k)=1.
\end{split}
\end{equation*}
Therefore, the {\sc MECSR} instance has a rule assignment to make $p$ being the possible winner of the election.

``$\Longleftarrow$'': Suppose there is a rule assignment of {\sc MECSR}, where the satisfaction of each voter $V_i$ is at least $d=1$. 
Let $R'$ be the set of rules of the rule assignment, $|R'|\leq k'$ and $r_{j'}\in R'$ is the rule assigned to $j-$th layer.
Then, for each voter $V_i$, it must hold ${\tt Sat}(V_i)=\sum_{j=1}^{k'}{\tt Sat}(v_i^j,r_{j'})\geq d=1$.
Since the satisfaction for a voter of each layer can only be $1$ or $0$, there must be a layer where the satisfaction is $1$ with the assigned rule $r_{j'}$.
It means:
\begin{equation*}
\begin{split}
     \exists r_{j'}\in R', {\tt Sat}(v_i^j,r_{j'})=1.
\end{split}
\end{equation*}
So, for each voter $V_i$, the corresponding vertex $v_i$ must hold $v_i \in [v_{j'}]$ with satisfying ${\tt Sat}(v_i^j,r_{j'})=1$.
Therefore, the set of vertices corresponding to the rules in $R'$ is a size-$\leq k'$ dominating set of~$\mathcal{G}$.
This completes the proof of this theorem.

\end{proof}

\begin{theorem}
\label{Sum-n and l}
The {\sc MECSR} problem with sum-model is NP-hard even if there are only two rules.
\end{theorem}
\begin{proof}
We prove the theorem by reducing from {\sc Dominating Set} problem.
We construct an {\sc MECSR} instance~$(E=(C,V,R,t), \alpha, d)$ from~$(\mathcal{G}=(\mathcal{V},\mathcal{E}), k')$ as follows.

For each vertex $v_i\in \mathcal{V}(i\in [m'])$, we construct a voter $V_i$.
We also construct another voter $V_{m'+1}$, $V=\bigcup_{i\in [m'+1]} V_{i}$.
There are two rules, $r_1$ and $r_2$, $R=\{r_1,r_2\}$, and $2m'$ layers, $t=2m'$.
The satisfaction of each vote is set as follows:
\begin{itemize}
    \item For the $j-$th layer with $j\in [m']$:
    \begin{itemize}
        \item ${\tt Sat}(v_i^j,r_1)=1$, $i\in [m']$ and $v_i\in N[v_j]$;
        \item ${\tt Sat}(v_i^j,r_1)=0$, $i\in [m']$ and $v_i\notin N[v_j]$;
        \item ${\tt Sat}(v_i^j,r_2)=0$, $i\in [m']$;
        \item ${\tt Sat}(v_i^j,r_1)=0$, ${\tt Sat}(v_i^j,r_2)=1$, $i= m'+1$.
    \end{itemize}
    \item For the $j-$th layer with $j\in [m'+1, 2m'-1]$:
    \begin{itemize}
        \item ${\tt Sat}(v_i^j,r_1)={\tt Sat}(v_i^j,r_2)=1$, $i\in [m']$;
        \item ${\tt Sat}(v_i^j,r_1)={\tt Sat}(v_i^j,r_2)=1$, $i=m'+1$ and $j\in [m'+1,m'+k']$;
        \item ${\tt Sat}(v_i^j,r_1)={\tt Sat}(v_i^j,r_2)=0$, $i=m'+1$ and $j\in [m'+k'+1,2m'-1]$;
    \end{itemize}
    \item For the $j-$th layer with $j=2m'$: 
    \begin{itemize}
    \item ${\tt Sat}(p,v_i^{2m'},r_1)={\tt Sat}(p,v_i^{2m'},r_2)=0$, $i\in [m'+1]$.
    \end{itemize}
\end{itemize}
Let $\alpha:=m'+1$, $d:=m'$. 
Now, we show that the {\sc Dominating set} instance has a size-$\leq k'$ dominating set if and only if there is a rule assignment of $r_1$ and $r_2$ of {\sc MECSR} instance with sum-model such that $\forall i\in [m'+1]$, $\sum_{j=1}^{2m'} {\tt Sat}(v_i^j) \geq m'$.

``$\Longrightarrow$'':
Suppose that there exists a size-$\leq k'$ dominating set $\mathcal{V'}$ in $\mathcal{G}$. 
For $j\in [m']$, $r_1$ is assigned to the $k'$ layers corresponding to the vertices in $\mathcal{V'}$, while $r_2$ is assigned to the other layers.
For example, if $v_2$ is in $V'$, $r_1$ is assigned to the second layer; otherwise, $v_2$ is not in $V'$, $r_2$ is assigned to the second layer.
For $j\in [m'+1,2m']$, the allocation of $r_1$ and $r_2$ is random.
Since $k'$ layers corresponding to the vertices in $\mathcal{V'}$ are assigned with rule $r_1$, and $m-k'$ layers are assigned with rule $r_2$, it holds:
\begin{itemize}
    \item When $j \in [m']$:
    \begin{itemize}
        \item $\sum\limits _{v_j \in V'}{\tt Sat}(v_i^j,r)=\sum\limits _{v_j \in V'}{\tt Sat}(v_i^j,r_1)+\sum\limits _{v_j \notin V'}{\tt Sat}(v_i^j,r_2) \geq 1+0=1$, $i \in [m']$;
        \item $\sum\limits _{v_j \in V'}{\tt Sat}(v_i^j,r)=\sum\limits _{v_j \in V'}{\tt Sat}(v_i^j,r_1)+\sum\limits _{v_j \notin V'}{\tt Sat}(v_i^j,r_2) = 0+(m'-k')=m'-k'$, $i = [m'+1]$.
    \end{itemize}
    \item When $j \in [m'+1,2m']$:
    \begin{itemize}
        \item $\sum\limits _{j \in [m'+1,2m']}{\tt Sat}(v_i^j,r)=\sum\limits _{j \in [m'+1,2m'-1]}{\tt Sat}(v_i^j,r)+{\tt Sat}(v_i^{2m'},r) = m'-1+0=m'-1$, $i \in [m']$;
        \item $\sum\limits _{j \in [m'+1,2m']}{\tt Sat}(v_i^j,r)=\sum\limits _{j \in [m'+1,m'+k]}{\tt Sat}(v_i^j,r)+\sum\limits _{j \in [m'+k'+1,2m']}{\tt Sat}(v_i^j,r) = k'+0=k'$, $i = m'+1$.
    \end{itemize}
    \item For the $j-$th layer with $j=2m'$: 
    \begin{itemize}
    \item ${\tt Sat}(p,v_i^{2m'},r_1)={\tt Sat}(p,v_i^{2m'},r_2)=0$, $i\in [m'+1]$.
    \end{itemize}
\end{itemize}
Note that, the satisfaction $v_i^j$ remains constant regardless of the $r_1$ or $r_2$ is assigned to $j-$th layer, when $j \in [m'+1,2m']$.
So, we can get:
\begin{itemize}
    \item ${\tt Sat}(V_i)=\sum\limits _{j \in [m']}{\tt Sat}(v_i^j,r)+\sum\limits _{j \in [m'+1,2m']}{\tt Sat}(v_i^j,r) \geq 1+m'-1=m'=d$, $i \in [m']$
    \item ${\tt Sat}(V_i)=\sum\limits _{j \in [m']}{\tt Sat}(v_i^j,r)+\sum\limits _{j \in [m'+1,2m']}{\tt Sat}(v_i^j,r) \geq 1+m'-1=m'=d$, $i = [m'+1]$
\end{itemize}
Therefore, the total satisfaction of $v_i^j$ is at least $m'$ for all $i\in [m'+1]$, $j\in [m'+1,2m']$.

``$\Longleftarrow$'':
Suppose there is a rule assignment of {\sc MECSR}, where the satisfaction of each voter $V_i$ is at least $m'$.
According to the votes of $V_{m'+1}$, the total satisfaction of $v_{m'+1}^j$ is always $k'$ regardless of $r_1$ or $r_2$ is assigned to the $j$-th layer, $j\in [m'+1,2m']$. 
That is:
\begin{equation*}
\begin{split}
     &\sum\limits _{j \in [m'+1,2m']}{\tt Sat}(v_{m'+1}^j,r_1) =\sum\limits _{j \in [m'+1,2m']}{\tt Sat}(v_{m'+1}^j,r_2) \\
     &=\sum\limits _{j \in [m'+1,m'+k']}{\tt Sat}(v_{m'+1}^j,r_1) +\sum\limits _{j \in [m'+k'+1,2m']}{\tt Sat}(v_{m'+1}^j,r_2) \\
     &= k'+0=k'.
\end{split}
\end{equation*}
${\tt Sat}(p,v_{m'+1}^j,r_1)={\tt Sat}(p,v_{m'+1}^j,r_2)$.
To reach the threshold $d=m'$, for the $j$-th layer, $j \in [m']$, at most $k'$ layers can be assigned with $r_1$. 
For the voters of $V_{i}$, $i\in [m']$, it always holds:
\begin{itemize}
    \item $\sum\limits _{j \in [m'+1,2m']}{\tt Sat}(v_i^j,r)+\sum\limits _{j \in [m'+1,2m'-1]}{\tt Sat}(v_i^{2m'},r) = m'-1$, $i \in [m']$
    \item $\sum\limits _{j \in [m'+1,2m']}{\tt Sat}(v_i^j,r)=\sum\limits _{j \in [m'+1,m'+k']}{\tt Sat}(v_i^j,r)+\sum\limits _{j \in [m'+k'+1,2m']}{\tt Sat}(v_i^j,r) \\= k'+0=k'$, $i = [m'+1]$
\end{itemize}
Therefore, for each $j$-th layer with $j\in [m']$, at least one layer is assigned with $r_1$ to receive the satisfaction of $1$ for $V_i(i \in [m'])$, and at most $k'$ layers can be assigned with $r_1$ to ensure that the total satisfaction $\sum\limits _{j \in [m']}{\tt Sat}(v_{m'+1}^j,r) \geq m'-k'$.
Since the $j$-th layer ($j\in [m']$) corresponds to the vertex $v_j$ in the graph $\mathcal{G}$, we have ${\tt Sat}(v_i^j,r_1)=1$ only when the corresponding vertex $v_i\in N[v_j]$, where $N[v_j]$ represents the neighborhood set of $v_j$.
Therefore, the instance $(E=(C,V,R,t), \alpha, d)$ has a solution of rule assignment if and only if the graph $(\mathcal{G}=(\mathcal{V},\mathcal{E}), k')$ has a size-$\leq k'$ dominating set.

\end{proof}

\begin{theorem}
\label{sum-alpha}
The {\sc MECSR} problem with sum-model is W[1]-hard with respect to the number of satisfied voters $\alpha$.
\end{theorem}
\begin{proof}
We prove this theorem by giving a reduction from {\scshape 3-Set Packing}, which given a set of elements $X$, $X=\bigcup_{i\in [m']}x_{i}$, a set of sets $\mathcal{S}$, $\mathcal{S}=\bigcup_{i\in [n']}\mathcal{S}_i$, where $\forall i\in [n']$,$\mathcal{S}_i\subset X$, $|\mathcal{S}_i|=3$, and asks for a size-$k'$ subset $\mathcal{S}'$ of $\mathcal{S}$ that $\forall \mathcal{S}_i, \mathcal{S}_{i'}\in \mathcal{S}'$, $\mathcal{S}_i\cap \mathcal{S}_{i'}=\emptyset$.
It is known that {\scshape 3-set packing} is NP-hard and is W[1]-hard with respect to $k'$.
We construct an {\sc MECSR} instance~$(E=(C,V,R,t), \alpha, d)$ from~$((X, \mathcal{S}), k')$ where $|X|=m'$ and $|\mathcal{S}|=n'$ as follows:

For each element $x_i\in X$, we construct a voter $V_i$, $V=\bigcup_{i\in [m']}V_i$.
There are totally $k'$ layers, $t=k'$, $V_i=\bigcup_{j\in [k']}v_i^j$.
For each set $\mathcal{S}_i\in \mathcal{S}$, we construct a rule $r_k$, $R=\bigcup_{k\in [n']}\{r_i\}$.
For each vote $v_i^j (j\in [k'])$ and rule $r_k$, if the corresponding element $x_i\in \mathcal{S}_k$, ${\tt Sat}(v_i^j,r_k)$ is set to $1$; otherwise, ${\tt Sat}(v_i^j,r_k)$ is set to $0$.
$$
{\tt Sat}(v_i^j,r_k)=\left\{
\begin{aligned}
&&1&,& x_i \in S_k, \\
&&0 &,& x_i \notin S_k. \\
\end{aligned}
\right.
$$
Note that the $k'$ votes of one voter are all the same.
Let $\alpha:=3k'$, $d:=1$.
Now we prove that there is a size-$k'$ subset $\mathcal{S}'$ if and only if there is a solution of {\sc MECSR} instance with sum-model.

``$\Longrightarrow$'': Suppose that there is a size-$k'$ subset $\mathcal{S}'\subset \mathcal{S}$, where $\forall \mathcal{S}_i, \mathcal{S}_{i'}\in \mathcal{S}'$, $\mathcal{S}_i\cap S_{i'}=\emptyset$.
Let $\mathcal{S}'=\{\mathcal{S}_{\beta1}, \mathcal{S}_{\beta2},\cdots, \mathcal{S}_{\beta k'}\}$.
For the $j-$th layer, the rule $r_{\beta j}$ is assigned.
Let $R'$ be the set of rules corresponding to the sets in $S'$.
In this way, $\forall x_i\in \mathcal{S}_{\beta j}\subset \mathcal{S}'(j\in [k'])$, the satisfaction of vote $v_i^j$ with rule $r_{\beta j}$ is $1$, ${\tt Sat}(v_i^j,r_{\beta j})=1$.
So, the satisfaction of the voter $V_i$ which corresponds to the element $x_i$ in $\mathcal{S}'$ is at least $1$, that is, ${\tt Sat}(p,V_i)\geq {\tt Sat}(p,v_i^j,r_{\beta j})\geq 1=d$.
There are exactly $3k'$ elements in the set of $\mathcal{S}'$, because each two sets in $\mathcal{S}'$ do not share a common element and each set contains exactly $3$ elements, $|\mathcal{S}_i|=|\mathcal{S}_{i'}|=3$.
That is, there are at least $\alpha=3k'$ voters whose satisfaction is at least $d=1$.
Therefore, if $((X, \mathcal{S}), k')$ has a size-$k'$ subset $\mathcal{S}'$ of $\mathcal{S}$ satisfying $\forall \mathcal{S}_i, \mathcal{S}_{i'}\in \mathcal{S}'$, $\mathcal{S}_i\cap \mathcal{S}_{i'}=\emptyset$, {\sc MECSR} instance has a solution of rule assignment $R'$.

``$\Longleftarrow$'': Suppose that there is an assignment of rules for each layer in which there are at least $\alpha=3k'$ voters whose satisfaction is at least $d=1$.
Let $r_{\gamma j}\in R$ be the rule assigned to $j-$th layer, $j\in [k']$.
For each satisfied voter $V_i$, it must hold ${\tt Sat}(V_i)=\sum_{j\in [k']}{\tt Sat}(v_i^j, r_{\gamma j})\geq 1=d$.
It means that there are at least one vote $v_i^j$ satisfying ${\tt Sat}(v_i^j, r_{\gamma j})=1$.
Without loss of generality, let the $j'-$th layer satisfies ${\tt Sat}(v_i^{j'}, r_{\gamma j'})=1$, $(j'\in [k'])$.
In this way, the corresponding element $x_i$ must be in the set $S_{\gamma j'}$.
Since there are $k'$ layers and at least $\alpha=3k'$  satisfied voters, the corresponding elements must be in the $k'$ corresponding sets.
Because each set $\mathcal{S}_i$ contains exactly $3$ elements, $\forall i\in [n'], |\mathcal{S}_i|=3$.
So, there are exactly $3k'$ elements in exact $k'$ sets, denote the set of the $k'$ sets as $S'$.
Therefore, if {\sc MECSR} instance has a solution of rule assignment, $((X, \mathcal{S}), k')$ has a size-$k'$ subset $\mathcal{S}'$ of $\mathcal{S}$ satisfying $\forall \mathcal{S}_i, \mathcal{S}_{i'}\in \mathcal{S}'$, $\mathcal{S}_i\cap \mathcal{S}_{i'}=\emptyset$.
This completes the proof of this theorem.
\end{proof}

Next, we continue consider the parameterized complexity of {\sc MECSR} with sum-model with the number of voters $n$ as parameter.
When the satisfaction of each vote is set dichotomous, 0 or 1, according to the proof of Theorem~\ref{Sum-t}, we can do some modifications on the constructions of layers and get an intractable result when $n\leq t$.
The other condition when $n> t$, we get a tractable result by constructing an ILP.
Otherwise, we find out that even there are only two votes, it is NP-hard to find out an acceptable rule assignment.
So, we get the following theorem.

\begin{theorem}
\label{sum-n}
The {\sc MECSR} problem with sum-model is NP-hard even if there are only 2 voters, $n=2$.
When the satisfaction of each vote is set dichotomous, 0 or 1, the {\sc MECSR} problem with sum-model is W[2]-hard when $n\leq t$, and is FPT when $n> t$, with the number of voter $n$ as parameter.
\end{theorem}
\begin{proof}
We firstly prove the result that {\sc MECSR} problem with sum-model is NP-hard even if there are only 2 voters, $n:=2$, by reducing {\sc Partition} problem to {\sc MECSR}. 
Given a set $X$ of $n'$ elements, where each element $x_i \in X$ is associated with a value $s_i \in S$, the problem asks for a partition of elements in $X$ into two disjoint subsets $X_1$ and $X_2$, $X=X_1\cup X_2$, $X_1\cap X_2=\emptyset$, such that the sum of values assigned to the elements in $X_1$ is equal to the sum of values assigned to the elements in $X_2$, $\sum_{x_i\in X_1}s_i = \sum_{x_{i'}\in X_2}s_{i'}$.
It is well-known that {\sc Partition} problem is NP-hard~\cite{ref_book26}.
We construct an {\sc MECSR} instance~$(E=(C,V,R,t), \alpha, d)$ from~$({X}, S)$ as follows:

We construct two voters $V_1$, $V_2$ and two rules $r_1$,  $r_2$, $V=V_1 \bigcup V_2$, $R=\{r_1, r_2\}$. 
There are $n'$ layers in total, $t:= n'$. 
For each layer of voter $V_1$, the satisfaction of vote $v_1^j(j\in [n'])$ with $r_1$ is set to $s_j$, and the satisfaction of vote $v_1^j$ with $r_2$ is set to 0. 
For each layer of voter $V_2$, the satisfaction of vote $v_2^j$ with rule $r_1$ is set to 0, and the satisfaction of vote $v_2^j$ with rule $r_2$ is set to $s_j$. 
$$
{\tt Sat}(v_i^j,r_k)=\left\{
\begin{aligned}
&&s_i &,& i = 1, r_k = r_1, j\in [k'], \\
&&0 &,& i = 1, r_k = r_2, j\in [k'], \\
&&0 &,& i = 2, r_k= r_1, j\in [k'], \\
&&s_i &,& i = 2, r_k = r_2, j\in [k'].\\
\end{aligned}
\right.
$$
Note that for the $j-$th layer, if the rule $r_1$ is chosen, then the satisfaction of $V_1$ is improved by $s_j$; otherwise, the rule of $r_2$ is chosen, and the satisfaction of $V_2$ is improved by $s_j$.
We set $\alpha:=2$, $d:= \frac{1}{2}N=\sum_{x_i\in X}s_i$.
It means each layer can be assigned with $r_1$ or $r_2$, corresponding to the assignment of elements to either $X_1$ or $X_2$; and $V_1$ and $V_2$ are both satisfied with a value $d=\frac{1}{2}N$, corresponding to partition $X$ into $X_1$, $X_2$, and the total value of elements in $X_1$ and $X_2$ are both $d=\frac{1}{2}N$.
Therefore, there is a partition for $(X,S)$ if and only if there is a solution of~$(E=(C,V,R,t), \alpha, d)$.

Next, we show the reason of why {\sc MECSR} problem is W[2]-hard with respect to the number of voter $n$ when $n\leq t$.
According to the proof of Theorem~\ref{Sum-t}, when $n=m'$ ($m'$ is the number of vertex) and $t=k'$ ($k'$ is the size of dominating set), {\sc MECSR} problem is w[2]-hard with respect to the number of layers $t$.
We can do the following modifications to the proof in Theorem~\ref{Sum-t}:
\begin{itemize}
    \item Add $\lambda$ layers for each voter($\lambda \geq m'-k'$);
    \item Set the satisfaction of vote $v_i^j$ with each rule is 0, ${\tt Sat}(v_i^j,r)=0, i\in [m'], j\in [k'+1, k'+\lambda], r\in R$. 
\end{itemize}
Let $R':=R, d':=d, t'=t+\lambda \geq m'$, and $(E'=(C',V',R',t'), \alpha', d')$ be the modified {\sc MECSR} instance.
It holds $n=m'\leq k'+\lambda=t'$.
Since the added layers have no influence on the the solution of this problem (the satisfaction of each added layers is 0).
So, the modified instance $(E'=(C',V',R',t'), \alpha', d')$ has a solution if and only if the original instance $(E=(C,V,R,t), \alpha, d)$ has a rule assignment solution.
This means {\sc MECSR} is w[2]-hard with respect to the number of $t$ when $n\leq t$.
Since if a problem is FPT with respect to $n$, this problem must be FPT with repect to $t$ when $n\leq t$.
Therefore, {\sc MECSR} problem is W[2]-hard with respect to the number of $n$ when $n\leq t$.

In the following, we show the FPT result when $n>t$.
Firstly, we can enumerate all conditions whose satisfaction achieve the threshold $d$.
The number of all conditions is $2^n$, and we consider the conditions of voters when the number of satisfied voter is at least $\alpha$.
For each of this condition, we construct an ILP formulation.
Let $V'$ be the set of voter whose satisfaction is
at least $d$.
We say two rules are of the same type if they can make the same voters satisfied in $j-$th layer($j\in [t]$).
Let $RT$ be the set of all rule types.
There are $n$ voters and $t$ layers in total, so there are at most $2^n\times t$ different rule types, $|RT|\leq 2^n\times t$.
For each rule type $rt\in RT$, let $n_{rt}$ be the number of rules in $R$ of type $rt$.
Let $f(j, rt)$ be the set of index of the voters whose satisfaction of vote $v_{i}^{j}$ is 1 with the rule $r$ of type $rt$, $i\in f(j,rt)$, ${\tt Sat}(v_{i}^{j},r)=1$.
If $i\in f(j,rt)$, we define $h(i,j,rt)=1$; otherwise $h(i,j,rt)=0$.
In the following, we define the variables of ILP.
For each rule type $rt$, we define an integer variable $x_{j,rt}$, where $x_{j,rt}\in \{0,1\}$ that $x_{j,rt}=c$ means there are $c$ rules of type $rt$ assigned to $j-$th layer.
The ILP instance consists of the following constraints:
\begin{displaymath} 
\begin{aligned}
1.\ &\sum_{rt\in RT}x_{j,rt}= 1, j\in [t]; \nonumber\\
2.\ &\sum_{rt\in RT}\sum_{j\leq [t]}x_{j, rt}\times h(i,j, rt)\geq d, \forall v_i\in V';\nonumber\\
3.\ &x_{j,rt}=0, 1, j\in [t], \forall rt\in RT. \nonumber\\
\end{aligned}
\end{displaymath}
The first equality guarantees that for each layer, there is exactly one rule assigned to this layer.
The second inequality means the chosen rules can make all voters in $V'$ satisfied.
Therefore, the solution of the ILP instance gives a rule assignment to make all voter in $V'$ satisfied.
The number of the variable is in $O(2^n\times t\times t)$.
For each condition, we construct such an ILP and there are totally $2^n$ conditions.
Since $n>t$, we can solve this problem in $O^*(2^n\times 2^n\times n\times n)$.
Therefore, the {\sc MECSR} problem is FPT with respect to $n$ when $n>t$.
This completes the proof of this theorem.
\end{proof}

\subsection{Complexity with Max-model}
In this section, we show the complexity results of {\sc MECSR} problem with max-model.
In max-model, the satisfaction of a voter is the maximal satisfaction from all $t$ votes, ${\tt Sat}(V_i)=max\{{\tt Sat}(v_i^j)\}(j\in [t])$.
Therefore, by comparing the satisfaction value $s$ of each vote to the threshold $d$, we can assign the satisfaction value as follows: if $s \geq d$, the satisfaction value is set to $1$; if $s < d$, the satisfaction value is set to $0$. Here, we set $d=1$.
Therefore, according to Theorem~\ref{Sum-t} and Theorem~\ref{sum-alpha}, we can directly obtain the following results: When the satisfaction of each vote is set to either $1$ or $0$, and the threshold $d$ is set to $1$, the {\sc MECSR} problem with the max-model has a solution if and only if the {\sc MECSR} problem with the sum-model has a solution.
\begin{theorem}
\label{max-t}
The {\sc MECSR} problem with max-model is NP-hard, is W[2]-hard with respect to the number of layers $t$.
\end{theorem}

\begin{theorem}
\label{max-alpha}
The {\sc MECSR} problem with max-model is W[1]-hard with respect to the number of satisfied voters $\alpha$.
\end{theorem}

In the following, we continue to analyze the effect of the number of rules $\ell$, the number of $\alpha$, the number of voters $n$ on the complexity of {\sc MECSR} problem with max-model, and get the following result.

\begin{theorem}
\label{max-l}
The {\sc MECSR} problem with max-model is NP-hard even if there are only two rules.
\end{theorem}
\begin{proof}
We prove this theorem by giving a reduction from {\scshape 3-SAT} problem, which given a set of boolean variables $X$ and a set of clauses $\mathcal{C}$ where each clause $\mathcal{C}_i\in \mathcal{C}$ is of the form: $\mathcal{C}_i=x_j\vee x_{j'}\vee x_{j''}$ with $x_j, x_{j'}, x_{j''}\in \{x, \overline{x}\}(x\in X)$, and asks for an assignment of all variables to makes all clauses true. 
It is well-known that $3$-SAT problem is NP-hard.
We construct an {\sc MECSR} instance~$(E=(C,V,R,t), \alpha, d)$ from~$(X,\mathcal{C})$ where $|X|=m'$ and $|\mathcal{C}|=n'$ as follows.

For each clause $\mathcal{C}_i \in \mathcal{C}(i \in [n'])$, we construct a voter $V_i$, $V=\bigcup_{i \in [n']} V_i$.
There are $m'$ layers in total, $t$ := $m'$. 
For each voter $V_i$, we construct $m'$ votes $V_i^j(j \in [m'])$, which are corresponding to $m'$ boolean variable $x_j \in X$ one by one. 
There are two rules $r_1$ and $r_2$, $R=\{r_1, r_2\}$.  
For each vote $v_i^j$, if the boolean variable ${x_j}$ occurs in clause $\mathcal{C}_i$, the satisfaction of $v_i^j$ is set to 1 with $r_1$ and is set to 0 with $r_2$, if the boolean variable $\overline{x_j}$ occurs in clause $\mathcal{C}_i$, the satisfaction of $v_i^j$ is set to 0 with $r_1$ and is set to 1 with $r_2$, if neither of $x_j$ and $\overline{x_j}$ occur in $\mathcal{C}_i$, the satisfaction of $v_i^j$ is set to 0 with $r_1$ or $r_2$.
$$
{\tt Sat}(v_i^j,r)=\left\{
\begin{aligned}
&&1 &,& r=r_1, x_j \in \mathcal{C}_i\ {\tt or}\ r=r_2, \overline{x_j}\in \mathcal{C}_i, \\
&&0 &,& r=r_1, \overline{x_j} \in \mathcal{C}_i \ {\tt or}\ r=r_2, x_j\in \mathcal{C}_i, \\
&&0 &,&x_j\notin \mathcal{C}_i \ {\tt and}\ \overline{x_j}\notin \mathcal{C}_i. 
\end{aligned}
\right.
$$
It means that if setting $x_j$ to true makes $\mathcal{C}_i$ true, then ${\tt Sat}(v_i^j,r_1)=1$, ${\tt Sat}(v_i^j,r_2)\\=0$; if setting $x_j$ to false makes $\mathcal{C}_i$ true, then ${\tt Sat}(v_i^j,r_2)=1$, ${\tt Sat}(v_i^j,r_1)=0$; if the value of $x_j$ have no influence on $\mathcal{C}_i$, then ${\tt Sat}(v_i^j,r_1)={\tt Sat}(v_i^j,r_2)=0$.
Let $d := 1$, $\alpha:=n'$.
Now we prove that there is an assignment of each boolean variable in $X$ to make all clauses in $\mathcal{C}$ true for $(X,\mathcal{C})$ if and only if there is a rule assignment of {\sc MECSR} with max-model.

``$\Longrightarrow$": Suppose there is an assignment of all variables which makes all clauses true.
For the $j-$th layer($j\in [m']$), if the corresponding boolean variable $x_j$ is set to true (or $1$), $r_1$ is assigned to this layer; otherwise, if the corresponding variable $x_j$ is set to false (or $0$), $r_2$ is assigned to this layer.
The rule assignment is denoted as $R'$.
Since the assignment of $X$ can make all clause true, for each clause $\mathcal{C}_i=x_j\vee x_{j'}\vee x_{j''}(x_j, x_{j'}, x_{j''}\in \{x, \overline{x}\}, x\in X)$, at least one of $x_j$, $x_{j'}$ and $x_{j''}$ is true.
Without of generality, we say $x_{k}(k \in \{j, {j'}, {j''} \})$ is true. If $x_{k}$ is in the form of $x$,
Then, the rule $r_1$ is assigned to this layer, and the satisfaction of $v_i^k$ must be $1$. Otherwise, if $x_{k}$ is true with the form of $\overline{x}$, then the rule $r_2$ is assigned to this layer, and the satisfaction of $v_i^{k}$ is $1$ as well.
Since all clause are true according to the assignment of $X$, for each voter $V_i$ (corresponding to each clause), it holds:
\begin{equation*}
\begin{split}
     {\tt Sat}(V_i)=max\{{\tt Sat}(v_i^j)\}\geq {\tt Sat}(v_i^k,r)=1=d.
\end{split}
\end{equation*}
Therefore, $R'$ is a solution of {\sc MECSR} with max-model.

``$\Longleftarrow$": Supposed there is a rule assignment $R'$ of {\sc MECSR}.
Since $\alpha=n'=|V|$, the satisfaction of each voter is at least $d=1$.
It means, for each voter $V_i$, it holds ${\tt Sat}(V_i))=max\{{\tt Sat}(v_i^j)\}\geq d=1$.
For the $j-$th layer($j\in [m']$), if $r_1$ is assigned to this layer in $R'$, the corresponding boolean variable $x_j$ is set to true; otherwise, $r_2$ is assigned to this layer in $R'$, the corresponding boolean variable is set to false.
Let $X'$ be the boolean variables assignment.
Since it holds ${\tt Sat}(p,V_i))=max\{{\tt Sat}(v_i^j)\}\geq d=1$, there are at least one layer $j'$, the satisfaction of $v_i^{j'}$ is $1$.
For each voter $V_i$, whose corresponding clause is $\mathcal{C}_i=x_j\vee x_{j'}\vee x_{j''}(x_j, x_{j'}, x_{j''}\in \{x, \overline{x}\}, x\in X)$, only the $j-$th, $j'-$th, and $j''-$th layers can receive a satisfaction of value $1$.
Without of generality, we set the satisfaction of vote $v_i^k(k \in \{j, {j'}, {j''} \})$ to $1$. 
That is:
\begin{itemize}
    \item If the satisfaction of vote $v_i^k$ with the rule $r_1$ is $1$, then $x_k$ is the form of $x$ and $x_k$ is true, so the clause $\mathcal{C}_i$ must be true;
    \item If the satisfaction of vote $v_i^k$ with the rule $r_2$ is $1$, $x_k$ is the form of $\overline{x}$ and $x_k$ is false, therefore the clause $\mathcal{C}_i$ must be true. . 
\end{itemize}
So, the corresponding clause $\mathcal{C}_i$ of voter $V_i$ must be true.
Therefore, $X'$ is a solution of $(X,\mathcal{C})$ to make all clause true.
This completes the proof of this theorem.
\end{proof}

Next, we continue consider the parameterized complexity of {\sc MECSR} with max-model with the number of voters $n$ as parameter.
According to the proof of theorem~\ref{sum-n}, we can change the second inequality to $\sum_{rt\in RT}\sum_{j\leq [t]}x_{j, rt}\times h(j, rt)\geq 1, \forall v_i\in V'$, to make all voters in $V'$ satisfied.
So, we get the following theorem.

\begin{theorem}
\label{max-n}
The {\sc MECSR} problem with max-model is W[2]-hard when $n\leq t$, and is FPT when $n> t$, with the number of voter $n$ as parameter.
\end{theorem}

\subsection{Complexity with Min-model}
In this section, we show the complexity results of {\sc MECSR} problem with min-model.
In min-model, the satisfaction of a voter is the minimal satisfaction from all $t$ votes, ${\tt Sat}(V_i)=\min\{{\tt Sat}(v_i^j)\},j\in[t]$.
When all voters accept the winner with min-model ($\alpha=n$), we can obviously find out an acceptable rule assignment in polynomial-time, that is, examining all $\ell$ rules to check whether the rule can make the satisfaction of all votes reach $d$.
It runs in $O(n\times t\times \ell)$ time.
So, {\sc MECSR} problem with min-model is in P when $\alpha=n$.
So, we continue consider the condition where $\alpha < n$.
For the parameterized complexity with the number of voters $n$ as parameter, we can enumerate all conditions which voters are satisfied.
There are totally $O(2^n)$ conditions.
For each condition, we just need to make the satisfaction of each votes reach $d$.
So, we can get a rule assignment or there is no such acceptable rule assignment in $O(2^n\times n\times t\times \ell)$ time.
Therefore, {\sc MECSR} is FPT with the number of voters $n$ as parameter.
In the following, we show the details of intractable and tractable results when with the number of satisfied voters $\alpha$, the number of layers $t$, or the number of rules $\ell$ as parameter.

\begin{theorem}
\label{min t-alpha}
The {\sc MECSR} problem with min-model is NP-hard when $\alpha\leq n$, and is W[1]-hard with respect to the number of satisfied voters $\alpha$ and the number of layers $t$.
\end{theorem}
\begin{proof}
We prove this theorem by reducing from {\scshape Multi-color Clique} problem, in which we are given a multi-color graph $\mathcal{G}=(\mathcal{V},\mathcal{E})$, where each vertex is assigned with a color. The graph $\mathcal{G}$ has a total of $k'$ colors, and $q$ vertices with the same color. Additionally, no two vertices with the same color are adjacent in the graph.
The aim is to find out a size-$k'$ clique $CL$ that each two vertices in $CL$ are adjacent.
It is known that {\scshape Multi-color Clique} problem is NP-hard and is W[1]-hard with respect to the clique size $k'$.
We construct an {\sc MECSR} instance~$(E=(C,V,R,t), \alpha, d)$ from~$(\mathcal{G}=(\mathcal{V},\mathcal{E}), k')$ where each vertex is denoted as $v_{i,i'}$ meaning the $i'-$th vertex with the $i-$th color, $\mathcal{V}=\bigcup_{i\in [k'], i'\in [q]}v_{i,i'}$, $|\mathcal{V}|=q\times k'$ and $|\mathcal{E}|=n'$.

The main ideals of the constructions are as follows:
\begin{itemize}
    \item Construct a voter for each vertex;
    \item Construct a layer for each color;
    \item Construct a rule for the index of each vertex of a color.
\end{itemize}
Now, we show the detail of the constructions.
For each vertex $v_{i,i'}\in V(i\in [k'], i'\in [q])$, we construct a voter $V_{\sigma}$ with $ \sigma=q \times (i-1)+i'$, $V=\bigcup_{\sigma=1}^{k'\times q} V_i$.
There are $q$ rules in total, $R=\bigcup_{i \in [q]}r_i$.
There are $k'$ layers in total, $t:=k'$.
For each voter $V_{\sigma}$, we construct $k'$ votes, $V_{\sigma}=\bigcup_{j \in [k']} v_{\sigma}^j$.
For each vote $v_{\sigma}^j$ and rule $r_k$ with $\sigma=q \times (i-1)+i'$, if the corresponding vertices $v_{i,i'}$ and $v_{j,k}$ satisfy $v_{i,i'}\in N[v_{j,k}]$, then ${\tt Sat}(v_{\sigma}^j,r_k)$ is set to $1$; otherwise, the ${\tt Sat}(v_{\sigma},r_k)$ is set to $0$.
$$
{\tt Sat}(v_{\sigma}^j,r_k)=\left\{
\begin{aligned}
&&1&,& v_{i,i'} \in N[v_{j,k}], \\
&&0 &,& v_{i,i'} \notin N[v_{j,k}]. \\
\end{aligned}
\right.
$$
Let $\alpha:=k'$, $d:=1$.
At most one rule can be assigned to each layer, corresponding to at most one vertex can be chosen for each color in $\mathcal{G}$; at least $\alpha=k'$ votes are needed to be satisfied, corresponding to at least $k'$ vertices be chosen to constitute a clique; the minimal satisfaction value of each votes for the satisfied voters is at least $d=1$, corresponding to each two chosen vertices are adjacent in $\mathcal{G}$.
Therefore, there is a size-$k'$ clique in $\mathcal{G}$ if and only if there is a rule assignment solution of~$(E=(C,V,R,t), \alpha, d)$.

\end{proof}

Next, we continue consider the parameterized complexity of {\sc MECSR} with min-model with the number of rules $\ell$ as parameter.
According to the proof of theorem~\ref{min t-alpha}, we can do some modifications on the constructions of layers and get an intractable result when $\ell< t$.
The other condition when $\ell\geq t$, we get a tractable result.
So, we get the following theorem.

\begin{theorem}
\label{min-ell}
The {\sc MECSR} problem with min-model is W[1]-hard when $\ell< t$, and is FPT when $\ell\geq t$, with the number of rules $\ell$ as parameter.
\end{theorem}

\section{Conclusion}
In this paper, we study the \emph{Multi-votes Election Control By Selecting Rules} problem, which allows each voter to present different votes in each layers among the set of candidates.
We study the computational complexity of this problem from the viewpoint of constructive control by assigning rules to each layer to make a special candidate $p$ being an acceptable winner of the election.
We find out that this problem is NP-hard for sum-model, max-model, or min-model.
Furthermore, we get the results that it is NP-hard even if there are only two voters in sum-model, or there are only two rules in sum-model or max-model; it is intractable with the number of layers as parameter for all of the three models, and even the satisfaction of each voter is set as dichotomous, either $1$ or $0$, it is remains hard to find out an acceptable rule assignment.
We also get some other tractable and intractable results, including fixed-parameter tractable, W[1]-hard and W[2]-hard.

For the future work, at first, we just consider the constructive cases here, the destructive control case may be a meaningful work. 
And, it is interesting to analyze the complexity of making special candidate $p$ being the unique winner of the election.
This work needs to consider the satisfaction of each candidate and the format of votes, which are ignored in this paper.
And, it is also interesting to make a fixed-size set of candidates being an acceptable committee with other rules, such as PAV, CCAV, and SAV for approval voting.
Another promising directing for future work is to embed the uncertainty rules to other models, such as the iterative elections.






\end{document}